\documentclass{article}

\usepackage{amssymb}
\usepackage{graphicx}

\usepackage{xy}
\xyoption{matrix}
\xyoption{frame}
\xyoption{arrow}
\xyoption{arc}

\usepackage{ifpdf}
\ifpdf
\else
\PackageWarningNoLine{Qcircuit}{Qcircuit is loading in Postscript mode.  The Xy-pic options ps and dvips will be loaded.  If you wish to use other Postscript drivers for Xy-pic, you must modify the code in Qcircuit.tex}
\xyoption{ps}
\xyoption{dvips}
\fi

\entrymodifiers={!C\entrybox}

\newcommand{\ket}[1]{{\left\vert{#1}\right\rangle}}

\newcommand{\braket}[2]{{\langle {#1}\!\mid\!{#2} \rangle}}
\newcommand{\Hilbert}{{\cal H}}

\newtheorem{definition}{Definition}[section]
\newtheorem{lemma}{Lemma}[section]

\newenvironment{proof}%
{\medskip
\noindent {\it Proof.}
}{$\Box$}

\begin{document}

\title{Quantum Hashing for Finite Abelian Groups}

\author{Alexander Vasiliev\thanks{Kazan Federal University}}


\date{}

\maketitle

\begin{abstract}
We propose a generalization of the quantum hashing technique based on the notion of the small-bias
sets. These sets have proved useful in different areas of computer science, and here their
properties give an optimal construction for succinct quantum presentation of elements of any finite
abelian group, which can be used in various computational and cryptographic scenarios. The known
quantum fingerprinting schemas turn out to be the special cases of the proposed quantum hashing for
the corresponding abelian group.
%
\end{abstract}

\section{Introduction}

Hashing is a necessary tool in a bag of tricks of every computer scientist. This term is believed
to be more than 60 years old and during its long history it has had a variety of useful
applications, which include cryptographic protocols, fast search, and data integrity check.

Recently, we have proposed a quantum version of this technique
\cite{Ablayev-Vasiliev:2014:Crypto-Q-Hashing}, which can also be useful in similar scenarios. For
instance, it is a suitable quantum one-way function that can be used in the quantum digital
signature protocol by Gottesman and Chuang \cite{GC:2001:Quantum-Digital-Signatures}. It can also
be used in different quantum computational models as a basis for efficient algorithms
\cite{Ablayev-Vasiliev:2014:QHashing-QOBDD} and communication protocols
\cite{Vasiliev:2015:QHashing-For-Communications}.


The classical hashing is deeply connected with error-correcting codes, i.e. as shown by Stinson
\cite{Stinson:1996:Hashing-Codes} they can be built from each other. The special case of
error-correcting codes called $\varepsilon$-balanced codes is related to another important
combinatorial object known as $\varepsilon$-biased sets \cite{Naor-Naor:1990:Small-Bias-Spaces},
which have applications in different areas of theoretical computer science, such as
derandomization, graph theory, number theory, etc. There are several known explicit constructions
of $\varepsilon$-balanced error-correcting codes \cite{Naor-Naor:1990:Small-Bias-Spaces},
\cite{Alon-et-al:1992:Almost-Independent-Random-Variables},
\cite{Ben-Aroya-Ta-Shma:2009:Small-Bias-Sets-from-Codes} that give rise to corresponding
$\varepsilon$-biased sets.

%
%
%
%
%
%
%
%

In this paper we show that $\varepsilon$-biased sets can be used to construct quantum hash
functions that have all the necessary cryptographic properties.

\section{Preliminaries}

The construction of quantum hashing in this paper relies on the notion of the $\varepsilon$-biased
sets. We use the definition given in
\cite{Chen-Moore-Russell:2013:Small-Bias-Sets-for-Nonabelian-Groups}.

Let $G$ be a finite abelian group and let $\chi_a$ be the characters of $G$, indexed by $a\in G$.

\begin{definition}
A set $S\subseteq G$ is called $\varepsilon$-biased, if for any nontrivial character $\chi_a$
\[
\frac{1}{|S|}\left|\sum\limits_{x\in S}\chi_a(x)\right|\leq\varepsilon.
\]
\end{definition}

It follows from the Alon-Roichman theorem \cite{Alon-Roichman:1994:Random-Cayley-graphs} that a set
$S$ of $O(\log|G|/\varepsilon^2)$ elements selected uniformly at random from $G$ is
$\varepsilon$-biased with high probability. The paper
\cite{Chen-Moore-Russell:2013:Small-Bias-Sets-for-Nonabelian-Groups} gives explicit constructions
of such sets thus derandomizing the Alon-Roichman theorem.


\section{Quantum Hashing}

Let $G$ be a finite abelian group with characters $\chi_a$, indexed by $a\in G$. Let $S\subseteq G$
be an $\varepsilon$-biased set for some $\varepsilon\in(0,1)$.

\begin{definition}\label{quantum-hash-function-definition}
We define a quantum hash function $\psi_{S}: G \to
(\Hilbert^2)^{\otimes\log{|S|}}$ as following:
\[
\ket{\psi_{S} (a)} = \frac{1}{\sqrt{|S|}}\sum\limits_{x\in S}\chi_a(x) \ket{x}.
\]
\end{definition}

The above function given an element $a\in G$ creates its quantum hash, which is a quantum state of
$\log{|S|}$ qubits. As mentioned earlier $S$ can be of order $O(\log|G|/\varepsilon^2)$, and thus
quantum hashing transforms its inputs into exponentially smaller outputs. That is, for any $a\in G$
represented by $\log|G|$ bits the number of qubits in its quantum hash would be
$\log{S}=O(\log\log|G|-\log\varepsilon)$.

The cryptographic properties of the hashing from Definition \ref{quantum-hash-function-definition}
are entirely determined by the $\varepsilon$-biased set $S\subseteq G$.

In particular all pairwise inner products of different hash codes (which is also the measure of
collision resistance \cite{Ablayev-Vasiliev:2014:Crypto-Q-Hashing}) are bounded by $\varepsilon$ by
the following Lemma.

\begin{lemma}
\[
\left|\braket{\psi_S(a_1)}{\psi_S(a_2)}\right|=
\frac{1}{|S|}\left|\sum\limits_{x\in S}\chi^*_{a_1}(x)\chi_{a_2}(x)\right|\leq\varepsilon,
\]
whenever $a_1\neq a_2$.
\end{lemma}
\begin{proof}
Let $\chi_{a_1}(x), \chi_{a_2}(x)$ be two different characters of $G$. Then $\chi^*_{a_1}(x)$ is
also a character of $G$, and so is the following function $\chi(x)=\chi^*_{a_1}(x)\chi_{a_2}(x)$.

$\chi(x)$ is nontrivial character of $G$, since $\chi_{a_1}(x)\not\equiv\chi_{a_2}(x)$ and
$\chi(x)=\chi^*_{a_1}(x)\chi_{a_2}(x)\not\equiv\chi^*_{a_1}(x)\chi_{a_1}(x)\equiv\textbf{1}$, where
$\textbf{1}$ is a trivial character of $G$.

Thus, Lemma follows from the definition of an $\varepsilon$-biased set
\[
\left|\braket{\psi_S(a_1)}{\psi_S(a_2)}\right|=
\frac{1}{|S|}\left|\sum\limits_{x\in S}\chi^*_{a_1}(x)\chi_{a_2}(x)\right|=\frac{1}{|S|}\left|\sum\limits_{x\in S}\chi(x)\right|\leq\varepsilon.
\]
\end{proof}

Irreversibility of $\psi_S$ is proved via the Holevo theorem and the fact that a quantum hash is
exponentially smaller than its preimage.

The size of the quantum hash above is asymptotically optimal because of the known lower bound by
Buhrman et al. \cite{Buhrman:2001:Fingerprinting} for the size of the sets of
pairwise-distinguishable states: to construct a set of $2^k$ quantum states with pairwise inner
products below $\varepsilon$ one will need at least $\Omega(\log(k/\varepsilon))$ qubits. This
implies the lower bound on the size of quantum hash of $\Omega(\log\log{|G|}-\log{\varepsilon})$.

In the next sections we give a more detailed look on the quantum hashing for specific finite
abelian groups. In particular, we are interested in hashing binary strings and thus it is natural
to consider $G=\mathbb{Z}^n_2$ 
and
$G=\mathbb{Z}_{2^n}$ (or, more generally, any cyclic group $\mathbb{Z}_{q}$).

\subsection{Hashing the Elements of the Boolean Cube}

For $G=\mathbb{Z}_2^n$ its characters can be written in the form
$\chi_a(x)=(-1)^{(a,x)}$, and quantum hash function is the following
\[
\ket{\psi_{S} (a)} = \frac{1}{\sqrt{|S|}}\sum\limits_{x\in S}(-1)^{(a,x)} \ket{x}.
\]

The resulting hash function is exactly the quantum fingerprinting by Buhrman et al.
\cite{Buhrman:2001:Fingerprinting}, once we consider an error-correcting code, whose matrix is
built from the elements of $S$. Indeed, as stated in
\cite{Ben-Aroya-Ta-Shma:2009:Small-Bias-Sets-from-Codes} an $\varepsilon$-balanced error-correcting
code can be constructed out of an $\varepsilon$-biased set. Thus, the inner product $(a,x)$ in the
exponent is equivalent to the corresponding bit of the codeword, and altogether this gives the
quantum fingerprinting function, that stores information in the phase of quantum states
\cite{Wolf:2001:PhD}.

%

\subsection{Hashing the Elements of the Cyclic Group}

For $G=\mathbb{Z}_q$ $\chi_a(x)=e^{\frac{2\pi i a x}{q}}$, and
quantum hash function is given by
\[
\ket{\psi_{S} (a)} = \frac{1}{\sqrt{|S|}}\sum\limits_{x\in S}e^{\frac{2\pi i a x}{q}} \ket{x}.
\]
The above quantum hash function is essentially equivalent to the one we have defined earlier in
\cite{Ablayev-Vasiliev:2014:Crypto-Q-Hashing}.

\subsubsection*{Acknowledgments.}

The work is performed according to the Russian Government Program of Competitive Growth of Kazan
Federal University. Work was in part supported by the Russian Foundation for Basic Research (under
the grants 14-07-00878, 15-37-21160).

\bibliography{references}

\begin{thebibliography}{10}

\bibitem{Ablayev-Vasiliev:2014:Crypto-Q-Hashing}
F~M Ablayev and A~V Vasiliev.
\newblock Cryptographic quantum hashing.
\newblock {\em Laser Physics Letters}, 11(2):025202, 2014.

\bibitem{GC:2001:Quantum-Digital-Signatures}
Daniel Gottesman and Isaac Chuang.
\newblock Quantum digital signatures.
\newblock Technical Report arXiv:quant-ph/0105032, Cornell University Library,
  Nov 2001.

\bibitem{Ablayev-Vasiliev:2014:QHashing-QOBDD}
Farid Ablayev and Alexander Vasiliev.
\newblock {Computing Boolean Functions via Quantum Hashing}.
\newblock In Cristian~S Calude, Rusins Freivalds, and Iwama Kazuo, editors,
  {\em Computing with New Resources}, Lecture Notes in Computer Science, pages
  149--160. Springer International Publishing, 2014.

\bibitem{Vasiliev:2015:QHashing-For-Communications}
Alexander Vasiliev.
\newblock Quantum communications based on quantum hashing.
\newblock {\em International Journal of Applied Engineering Research},
  10(12):31415--31426, 2015.

\bibitem{Stinson:1996:Hashing-Codes}
D.~R. Stinson.
\newblock On the connections between universal hashing, combinatorial designs
  and error-correcting codes.
\newblock In {\em In Proc. Congressus Numerantium 114}, pages 7--27, 1996.

\bibitem{Naor-Naor:1990:Small-Bias-Spaces}
Joseph Naor and Moni Naor.
\newblock Small-bias probability spaces: Efficient constructions and
  applications.
\newblock In {\em Proceedings of the Twenty-second Annual ACM Symposium on
  Theory of Computing}, STOC '90, pages 213--223, New York, NY, USA, 1990. ACM.

\bibitem{Alon-et-al:1992:Almost-Independent-Random-Variables}
Noga Alon, Oded Goldreich, Johan Hastad, and Rene Peralta.
\newblock Simple constructions of almost k-wise independent random variables.
\newblock {\em Random Structures \& Algorithms}, 3(3):289--304, 1992.

\bibitem{Ben-Aroya-Ta-Shma:2009:Small-Bias-Sets-from-Codes}
A.~Ben-Aroya and A.~Ta-Shma.
\newblock Constructing small-bias sets from algebraic-geometric codes.
\newblock In {\em Foundations of Computer Science, 2009. FOCS '09. 50th Annual
  IEEE Symposium on}, pages 191--197, Oct 2009.

\bibitem{Chen-Moore-Russell:2013:Small-Bias-Sets-for-Nonabelian-Groups}
Sixia Chen, Cristopher Moore, and Alexander Russell.
\newblock Small-bias sets for nonabelian groups.
\newblock In Prasad Raghavendra, Sofya Raskhodnikova, Klaus Jansen, and
  Jose~D.P. Rolim, editors, {\em Approximation, Randomization, and
  Combinatorial Optimization. Algorithms and Techniques}, volume 8096 of {\em
  Lecture Notes in Computer Science}, pages 436--451. Springer Berlin
  Heidelberg, 2013.

\bibitem{Alon-Roichman:1994:Random-Cayley-graphs}
Noga Alon and Yuval Roichman.
\newblock Random cayley graphs and expanders.
\newblock {\em Random Structures \& Algorithms}, 5(2):271--284, 1994.

\bibitem{Buhrman:2001:Fingerprinting}
Harry Buhrman, Richard Cleve, John Watrous, and Ronald de~Wolf.
\newblock Quantum fingerprinting.
\newblock {\em Phys. Rev. Lett.}, 87(16):167902, Sep 2001.

\bibitem{Wolf:2001:PhD}
Ronald de~Wolf.
\newblock {\em Quantum Computing and Communication Complexity}.
\newblock PhD thesis, University of Amsterdam, 2001.

\end{thebibliography}

\end{document}